\documentclass[11pt]{article}
\usepackage{geometry}
\usepackage{amsfonts,amsthm,amssymb,amsmath,amscd}
\usepackage{graphicx, float}
\usepackage{longtable}
\usepackage{ae}
\usepackage{fancybox}
\usepackage{multirow}
\usepackage{natbib}
\usepackage{listings}
\usepackage{rotating}
\usepackage[pdftex]{hyperref}
\usepackage{pdflscape}
\usepackage{eurosym}
\newcommand{\bs}[1]{ \boldsymbol{#1} }
\usepackage{here}
\usepackage{color}
\usepackage[normalem]{ulem}
\usepackage{a4wide}
\begin{document}

\newtheorem{theorem}{Theorem}[section]
\newtheorem{lemma}[theorem]{Lemma}
\theoremstyle{definition}
\newtheorem{definition}[theorem]{Definition}
\newtheorem{folgerung}[theorem]{Corollary}
\newtheorem{proposition}[theorem]{Proposition}
\newtheorem{keywords}[theorem]{Keywords}
\newtheorem{conclusion}[theorem]{Conclusion}
\newtheorem{question}[theorem]{Question}
\newtheorem{rem}[theorem]{Remark}
\newtheorem{example}[theorem]{Example}
\newcommand{\Var}{\mbox{Var}}
\newcommand{\Cov}{\mbox{Cov}}
\newcommand{\SCR}{\mbox{SCR}}
\newenvironment{bew}{\begin{proof}[Proof]}{\end{proof}}

\title{Parameter uncertainty and reserve risk under Solvency II}
\author{Andreas Fr\"ohlich und Annegret Weng}
\date{\today}
\maketitle

\begin{abstract}
In this article we consider the parameter risk in the context of internal modelling of the reserve risk under Solvency II.\\
We discuss two opposed perspectives on parameter uncertainty and point out that standard methods of classical reserving focusing on the estimation error of claims reserves are in general not appropriate to model the impact of parameter uncertainty upon the actual risk of economic losses from the undertakings's perspective.\\
Referring to the requirements of Solvency II we assess methods to model parameter uncertainty for the reserve risk by comparing the probability of solvency actually attained when modelling the solvency risk capital requirement based on the respective method to the required confidence level. Using the simple example of a normal model we show that the bootstrapping approach is not appropriate to model parameter uncertainty according to this criterion. We then present an adaptation of the approach proposed in \cite {froehlich2014}. 
Experimental results demonstrate that this new method yields a risk capital model for the reserve risk achieving the required confidence level in good approximation.
\end{abstract}
\section{Introduction}\label{sec:intro}
The Solvency II directive (\cite{solvency}) defines the capital requirement of an insurance undertaking as the value-at-risk of the loss of basic own funds for the confidence level $\alpha=99.5\%$ over a one-year time horizon (cf. \cite{solvency} Article 101). We interpret the change in basic own funds as a random variable.\\
\\
An effective risk management does not only require the consideration of the overall risk of an insurance undertaking, but also an assessment  of the material subrisks. If we interpret the loss of basic own funds over a one-year horizon due to a particular subrisk as a random variable $\bs X$, it is best practice to define the standalone risk capital requirement for this subrisk analogously to Article 101 in \cite{solvency} as the 99.5\% value-at-risk of $\bs X$.\\
\\
However, there is not only uncertainty about the future outcomes of $\bs X$ caused by random fluctuation, but also about the true distribution of $\bs X$. Therefore, the true 99.5\% value-at-risk of $\bs X$ is unknown and the insurance undertaking can only estimate its solvency capital requirement.\\
\\
In this article, we assume that the undertaking uses an internal model and estimates the parameters specifying $\bs X$ from historical data. The possible deviation of the parameter estimates from the true parameters causes parameter uncertainty. In the sequel we ignore the basic model uncertainty and concentrate on the parameter uncertainty.\\ 
In this situation there are two sources of uncertainty:
\begin{enumerate}
	\item the random variable $\bs X$,
	\item the uncertainty with respect to the modelled solvency capital requirement \textbf{SCR} resulting from the randomness of the historical data used to estimate the parameters.
\end{enumerate}
Note that due to 2. the modelled solvency capital requirement \textbf{SCR} itself is a random variable.\\
\\
From \cite{solvency}, Article 101 we derive the following question:
\begin{question}\label{qu:1}
How can we model the solvency capital requirement \textbf{SCR} (for a subrisk) such that it will not be exceeded by the possible loss $\bs X$ of basic own funds (due to this subrisk) over a one-year horizon with a probability of 99.5\% - taking into account the randomness of both $\bs X$ and \textbf{SCR}?
\end{question}
This question corresponds to a central idea of predictive inference, see e.g. \cite{barndorff, young,serverini}, and has been investigated in the context of Solvency II for several distributional assumptions for the random variable $\bs X$ in various articles (see e.g. \cite{gerrard,froehlich2014,bignozzi,bignozzi2}). Furthermore, Question 1.1 corresponds to an unbiased estimate of the value-at-risk in the sense of \cite{pitera}.\\
\\
In the sequel we restrict to the modelling of a standalone solvency capital requirement for the reserve risk in the sense of Question \ref{qu:1}. In the context of the reserve risk, $\bs X$ is the loss according to the one-year development result of incurred claims and the solvency capital requirement \textbf{SCR} is the standalone solvency capital requirement for the one-year reserve risk. We assume that we can write $\bs X=\bs X(\theta)$ for some fixed, but unknown parameter vector $\theta$. The undertaking does not know $\theta$ but can only derive an estimate $\hat{\theta}$ based on the observed claims development triangle $D$. These notions will be made precise in Section \ref{sec:def}. For the sake of simplicity we ignore the impact of interest rates and (the development result of) the risk margin.\\
\\
There is an extensive literature on the reserve risk dealing with the prediction error of claims reserves. However, the majority of these contributions does not consider the value-at-risk of the one-year claims development result, but investigates the ultimate mean squared prediction error, see e.g. \cite{Mack1993,Mack1999,wuethrich2008}. 
\cite{merz2008} investigate the one-year development result, but they also use the mean squared estimation error as the risk measure.\\
For Mack's chain ladder model the uncertainty about the volatility parameters $\sigma_k^2$ is often ignored by just replacing $\sigma_k^2$ by its estimate $\hat{\sigma}_k^2$ since this uncertainty has only minor impact on the mean squared estimation error of the claims reserve. However, the uncertainty about volatility parameters has a crucial impact on the value-at-risk of the corresponding predictive distribution according to the solvency capital requirement.\\
To model the predictive distribution the existing literature recommends either the Bayesian approach or bootstrapping, see e.g. \cite{bjoerkwall2011,england1999,england2006,gisler2006,pinheiro2003}. For the one-year risk we refer to \cite{krause2009}. However, none of these articles addresses Question \ref{qu:1}.\\ 
\\
This article is a first contribution how to model the solvency capital requirement for the reserve risk with respect to the required confidence level of 99.5\% in the sense of Question \ref{qu:1}. To understand the economic relevance of parameter uncertainty with respect to the unknown future losses, it is important to distinguish between the following two perspectives:
\begin{enumerate}
	\item The theoretical perspective: From the theoretical perspective both the historical data and the parameter estimate  $\hat{\bs \theta}$ are random. The true parameter vector is not known, but fixed.
	\item The undertaking's perspective: From this perspective, there is only one fixed sample of historical data. Thus, using a fixed estimation method, the estimate $\hat{\theta}$ is fixed. There is uncertainty about the true parameter $\theta$.
\end{enumerate}
Note that the uncertainty about the true parameter $\theta$ from the undertaking's perspective refers to the actual economic risk of potential true losses, which depends on $\theta$ (but is not directly affected by the estimate $\hat{\theta}$). This indicates that the undertaking's perspective yields the basis for an economic interpretation of parameter risk. In Section \ref{sec:perspective} we explain in more detail why the actual economic risk relevant for Solvency II is given by the undertaking's perspective reflecting the real situation of the undertaking. For illustration, we use a simple example demonstrating the difference between the theoretical perspective and the undertaking's perspective. For this example we show that the two perspectives lead to different parameter distributions and prove that indeed the parameter distribution corresponding to the undertaking's perspective yields an exact solution to Question \ref{qu:1}.  Moreover, we explain why standard methods of classical reserving based on the theoretical perspective are in general not appropriate to model the impact of parameter uncertainty upon the actual risk of economic losses.\\
\\
A possibility to model parameter uncertainty from the undertaking's perspective would be the application of the Bayesian approach. Note that these techniques use additional a-priori-information resp. expert judgement (cf. e.g. \cite{wuethrich2008} Chapter 4, \cite{verrall1990} or \cite{peters}).\\
However, in order to find a solution to Question \ref{qu:1} for the reserve risk we follow the approach introduced in \cite{froehlich2014} (in the sequel this method is called ``inversion method'') based on Fisher's idea of fiducial inference \cite{fisher1930}. While there was a lot of criticism of Fisher's original argument (see  \cite{zabell92} or \cite{hannig2016}, p.2), in the last decades many authors reinvestigated Fisher's idea and showed that fiducial inference, properly generalized, yields solutions to many important inference problems (see e.g. \cite{hannig2016}, \cite{iyer2004}, \cite{hannig2013}, \cite{hannig2006}, \cite{wang2012}).\\
\\
In \cite{froehlich2014} the authors proved that their approach to model parameter uncertainty based on fiducial inference yields an exact solution to the fundamental Question \ref{qu:1} in the context of Solvency II for a wide class of distributions. Thus, it is a straight forward idea to apply the inversion method introduced in \cite{froehlich2014} in order to solve Question 1.1 for the reserve risk in the context of Solvency II. Insofar, we concentrate on modelling a predictive distribution from the undertaking's perspective without using any a-priori-information or expert knowledge. In particular, applying the inversion method we avoid a sophisticated Markov Chain Monte Carlo simulation, usually necessary to perform a Bayesian analysis.\\ 
\\
Referring to Question \ref{qu:1} in Section \ref{sec:exmethods} we assess several methods to model the risk capital requirement for the reserve risk by investigating the probability of solvency (cf. Section \ref{subsec:angemessen}). To illustrate the effects of parameter uncertainty we consider a very simple model - the normal model (cf. \cite{england2006, gisler2006}). We discuss the bootstrapping approach and present experimental results demonstrating that even for this simple model bootstrapping is not appropriate in the sense of Question \ref{qu:1}, since it does not guarantee the required solvency level of 99.5\% under the consideration of the randomness of the historical data.\\
In Section \ref{subsec:inv} we adjust the inversion method proposed in \cite{froehlich2014} to derive a risk capital model for the reserve risk achieving the required probability of solvency in good approximation.

\section{Parameter uncertainty and reserve risk}\label{sec:def}

\subsection{Basic definitions}\label{subsec:grand}
Random variables are printed in \textbf{bold}.\\
Throughout the article, $t=0$ denotes the time corresponding to the current solvency balance sheet and $t=1$ denotes the end of the one-year period. In the sequel, we call a quantity ``unknown'' if it is unknown to the undertaking.\\

Let $\bs C_{i,k}$ denote the cumulative claims payments of accident year $i$, $0\le i\le n$, up to development year $k$, $0\le k\le n$. We interpret $\bs C_{i,k}$ as a random variable for which we observe realizations $C_{i,k}$ for $i+k\le n$. In the sequel $D$ denotes the observed claims development triangle $\{C_{i,k}: i+k\le n\}$ which is considered as a realization of a random vector $\bs D$. For the sake of simplicity, we neglect the effect from interest rates upon the best estimate reserve.\\
Using an appropriate reserving method to estimate the ultimate claims payment $\hat{C}_{i,n}$ of accident year $i=1,\ldots,n$, the best estimate reserve is given by
$
\hat{R}_0^i=\hat{C}_{i,n}-C_{i,n-i}
$
for all accident years $i=1,\ldots,n$. For simplicity we assume that all claims are settled after $n$ development years.\\
The total best estimate reserve is given by
$
\hat{R}_0=\sum_{i=1}^n \hat{R}^i_0.
$
We use the notation $\hat{R}_0(D)$ to emphasize the dependency of $\hat{R}_0$ on the realization $D$ of $\bs D$.\\
Consider the random payments of the next calendar year
\begin{equation}\label{eq:z}
\bs Z_{i,n-i+1}=\bs C_{i,n-i+1} - C_{i,n-i}
\end{equation}
for $1\le i\le n$ and set $\bs Z=\sum_{i=1}^n \bs Z_{i,n-i+1}$.\\
We denote the best estimate reserve at $t=1$ for the same accident years $0\le i\le n$ by $\hat{\bs R}_1$. As in \cite{krause2009,merz2008} we assume that $\hat{\bs R}_1$ is determined by the claims observed up to time $t=1$, i.e. the claims development triangle $D$ in $t=0$ extended by the diagonal representing the payments of the next calendar year $\overrightarrow{\bs Z}=(\bs Z_{i,n-i+1}:1\le i\le n)$, using an appropriate reserving method. We write $\hat{\bs R}_1=\hat{R}_1(D,\overrightarrow{\bs Z})$ to stress this deterministic dependency.\\
The one-year claims development loss $\bs X=\bs Z+\hat{\bs R}_1-\hat{R}_0(D)$ describes the possible loss caused by the difference between the best estimate reserve in $t=0$ and the sum of the expenses for the claims payments within the next year and the expenditure for setting up the reserve at the end of the next year. For simplicity we use $\bs S$ to denote $\bs Z+\hat{\bs R}_1$.\\
\\
Note that for many common reserving methods the random quantities described above can be written in the form $\bs Z_{i,n-i+1}=Z_{i,n-i+1}(\bs \zeta_{i,n-i+1},D,\theta_{n-i+1})$, $\bs Z=Z(\overrightarrow{\bs \zeta},D,\theta)$, $\bs S=S(\overrightarrow{\bs \zeta},D,\theta)$, $\bs X=X(\overrightarrow{\bs \zeta},D,\theta)$ with appropriate mappings $Z_{i,n-i+1}$, $Z$, $S$ and $X$ where $\theta$ is the parameter vector according to the chosen reserving method and $\overrightarrow{\bs \zeta}$ is a random vector of future standardized residues whose distribution is independent of $\theta$.\\
\\
As an example we consider the stochastic chain ladder model introduced by Mack (see \cite{Mack1993,Mack1999}): Let
$$
F_{i,k}=\frac{C_{i,k}}{C_{i,k-1}}.
$$
and assume that there exist factors $f_1,\ldots,f_n>0$ and variance parameters $\sigma_1^2,\ldots,\sigma_n^2$ such that for all $0\le i\le n$ and $1\le k\le n$
we have
\begin{itemize}
\item $E[\bs F_{i,k}| C_{i,0},\ldots, C_{i,k-1}]=f_k$,
\item $\Var[\bs F_{i,k}| C_{i,0},\ldots, C_{i,k-1}]=\frac{\sigma_k^2}{C_{i,k-1}^{\gamma}}$ for $\gamma=0$ or $\gamma=1$  and
\item independence of the accident years: the vectors $(\bs C_{i,0},\ldots,\bs C_{i,n})$, $0\le i\le n$, are independent.
\end{itemize} 
Recall that we assume that all claims are settled after $n$ years, i.e. $f_n=1$ and $\sigma_n^2=0$. We assume that the parameter vector $\theta=(f_1,\sigma_1^2,f_2,\sigma_2^2,\ldots,f_{n-1},\sigma_{n-1}^2)$ is unknown.\\
Unbiased estimates are given by 
\begin{align}\label{eq:estim}
\hat{f}_k&=\sum_{i=0}^{n-k} C_{i,k-1}^\gamma F_{i,k}/\sum_{i=0}^{n-k} C_{i,k-1}^\gamma,\nonumber\\
\hat{\sigma}_k^2&=\frac{1}{n-k}\sum_{i=0}^{n-k} C_{i,k-1}^\gamma \left(F_{i,k}-\hat{f}_k\right)^2
\end{align}
for $k=1,\ldots,n-1$. We set $\hat{f}_n=1$ and $\hat{\sigma}_n^2=0$. The estimated parameter vector is given by $\hat{\theta}=(\hat{f}_1,\hat{\sigma}_1^2,\ldots,\hat{f}_{n-1},\hat{\sigma}_{n-1}^2)$. In particular, the best estimate reserve $\hat{R}_0$ is determined by the chain ladder procedure using the chain ladder factors $(\hat{f}_1,\ldots,\hat{f}_{n-1},\hat{f}_{n})$.\\

Furthermore, we assume that for $0\le i\le n$, $1\le k\le n$ the individual chain ladder factors can be written as
\begin{equation}\label{eq:res}
\bs F_{ik}=f_k+\frac{\sigma_k}{C_{i,k-1}^\gamma}\cdot \bs \zeta_{ik}
\end{equation}
where $\bs \zeta_{ik}$ are iid. residues with mean 0 and variance 1 and $\gamma\in\{0,1\}$. Let $\overrightarrow{\bs \zeta}=(\bs \zeta_{i,n-i+1}: 1\le i\le n)$ be the random residue vector of the next business year and set $\theta_k:=(f_k,\sigma_k^2)$. In particular, based on the random payments $\bs Z_{i,n-i+1}=(\bs F_{i,n-i+1}-1)\cdot C_{i,n-i}$, $i=1,\ldots,n$, the best estimate reserve $\hat{\bs R}_1=\hat{R}_1(D,\overrightarrow{\bs Z})$ is determined using the chain ladder procedure based on the development triangle $D$ extended by the ``new diagonal'' $\overrightarrow{\bs Z}$.

\subsection{Modelled risk and probability of solvency}\label{subsec:angemessen}
In this subsection we introduce the notion of modelled risk, modelled risk capital and probability of solvency taking parameter uncertainty into account.\\
\\
Let $\bs X=\bs X(\theta)$ be a random variable describing a subrisk of the undertaking, whose distribution depends on an unknown parameter vector $\theta$. For the reserve risk consider the one-year claims development loss 
$$
\bs X=\bs S(\theta)-\hat{R}_0(D)
$$
where $\bs S(\theta)= \bs Z(\theta)+\hat{R}_1(D,\overrightarrow{\bs Z}(\theta))$ (cf. Section \ref{subsec:grand}).\\
If the parameter vector $\theta$ was known, the required risk capital for the one-year reserve risk for the confidence level $\alpha$ would just be the $\alpha$-quantile of the random variable  
$\bs X=X(\overrightarrow{\bs \zeta},D,\theta)$ (cf. the notation introduced in Subsection \ref{subsec:grand}).\\
But since the undertaking does not know the parameter vector $\theta$, it does not know the true distribution of $\bs X$. Hence, we assume that it can only calculate the risk capital requirement based on the observed historical data $D$, which is a realization of the random vector $\bs D$. We assume that $\bs D$ and $\overrightarrow{\bs \zeta}$ are independent.\\

Given the observed data $C_{i,k}$ for $i+k\le n$ we assume that the undertaking models its risk as a predictive distribution by the following two-step procedure:
\begin{enumerate}
\item Given a method $M$ and the triangle $D=\{C_{i,k}:i+k\le n\}$ generate a probability distribution $\mathcal{P}=\mathcal{P}(D;M)$ in order to simulate a random parameter vector $\bs \theta^{sim}$.
\item Consider the modelled claims development loss $\bs X^{model}:=X^{model}(\overrightarrow{\bs \zeta^\prime},D,\bs \theta^{sim})$ for an independent copy $\overrightarrow{\bs \zeta^\prime}$ of $\overrightarrow{\bs \zeta}$ (cf. the notation introduced in Subsection \ref{subsec:grand}).
\end{enumerate}
We assume that $\overrightarrow{\bs \zeta}$, $\overrightarrow{\bs \zeta^\prime}$ and $\bs \theta_{sim}$ are independent. The random variable $\bs X^{model}$ depends on the data $D$, but also on the method $M$ resp. the chosen parameter distribution $\mathcal{P}$. In practice, the procedure above is typically performed using a Monte-Carlo simulation.\\
For simplicity, we assume that the cumulative distribution function $F_{{\bs X}^{model}}$ of $\bs X^{model}$ is invertible.
\begin{definition}
	Let a method $M$ and a claims development triangle $D$ be given. Referring to the two step procedure described above we call $\bs \theta_{sim}$ the \textbf{modelled parameter (vector)} and the random variable $\bs X^{model}$ the \textbf{modelled risk}. For $0<\alpha<1$ we refer to
$$
\SCR(\alpha;D;M):=F_{\bs X^{model}}^{-1}(\alpha)
$$
as the \textbf{modelled risk capital with respect to the confidence level $\alpha$}. Taking the randomness of the historical data $\bs D$ into account we call
\begin{eqnarray}\label{eq:konkret}
P\left(\bs X\le \SCR(\alpha;\bs D;M)\right)=P\left(X(\overrightarrow{\bs \zeta},\bs D,\theta)\le \SCR(\alpha;\bs D;M)\right)
\end{eqnarray}
the \textbf{probability of solvency} according to the corresponding risk capital model subject to the method $M$.
\end{definition} 
Note that in (\ref{eq:konkret}) not only $X(\overrightarrow{\bs \zeta},\bs D,\theta)$, but also $\SCR(\alpha;\bs D;M)$ is considered to be random.
\begin{rem}\label{rem:parameter}
	We carefully distinguish between the modelled and the ``true'' quantities:
	\begin{itemize}
		\item The distribution of the modelled parameter vector $\bs \theta_{sim}$ depends on the observed data at time $t=0$ and on the choice of the method $M$. However, the ``true'' parameter vector $\theta$ is still assumed to be unknown but fixed.
		\item We stress the difference between the modelled risk $\bs X^{model}$ and the ``true'' risk $\bs X$: Note that $\bs X^{model}$ and $\bs X$ are independent random variables, that are, in general, not even from the same distribution family. Consider the time $t=0$ and let the data $D$ be given. At $t=0$ the distribution $\bs X^{model}$ has already been specified inside the risk capital model. In contrast, since $\theta$ is unknown, the probability distribution of $\bs X$ is unknown from the undertaking's perspective.
	\end{itemize}
\end{rem}
Using the notation introduced above we reformulate the central Question \ref{qu:1}: \begin{question}\label{qu:2}
	Given a confidence level $0<\alpha<1$. How can we determine a method $M$ to model a parameter distribution of $\bs \theta_{sim}$ such that
	\begin{equation}
	\label{eq:q2}
	P\left(\bs X\le SCR(\alpha;\bs D;M)\right)=\alpha ?
	\end{equation}	
\end{question}
We refer to this central question in order to assess methods $M$ to model parameter uncertainty in the context of Solvency II. However, for complex practical problems (as the reserve risk) it may be hard to find an exact solution for this question. In this case we aim for a method $M$ solving (\ref{eq:q2}) in good approximation.
\begin{rem}\label{rem:backtesting} Note that there is a close relation of the probability of solvency to backtesting
	(cf. \cite{gerrard}, p. 731 or \cite{froehlich2014}, Remark 2). 
\end{rem}
Applying the two-step procedure described above to the chain-ladder model we determine a distribution for the modelled parameter vector $$\bs \theta^{sim}=(\bs f_1^{sim},(\bs \sigma_1^{sim})^2,\bs f_2^{sim},(\bs \sigma_2^{sim})^2,\ldots,\bs f_{n-1}^{sim},(\bs \sigma_{n-1}^{sim})^2)$$ and model the incremental payment $\bs Z_{i,n-i+1}$ by $\bs Z^{model}_{i,n-i+1}=\bs C_{i,n-i+1}^{model}-C_{i,n-i}$ where
$
\bs C_{i,n-i+1}^{model}=C_{i,n-i}\cdot \bs F_{i,n-i+1}^{sim}
$
and
$$
\bs F_{i,n-i+1}^{sim}=\bs{f}_{n-i+1}^{sim}+\frac{\bs \sigma_{n-i+1}^{sim}}{\sqrt{C_{i,n-i}^\gamma}}\cdot \bs \zeta_{i,n-i+1}^\prime
$$
for a random parameter vector $\bs \theta_{n-i+1}^{sim}=(\bs f_{n-i+1}^{sim},(\bs \sigma_{n-i+1}^{sim})^2)$, $1\le i\le n$, determined by some method $M$ and independent of $\overrightarrow{\bs \zeta}^\prime=\left\{\bs \zeta_{i,n-i+1}^\prime:1\le i\le n\right\}$ with iid. modelled residues $\bs \zeta_{i,n-i+1}^\prime$.\\
We model the reserve $\hat{R}_1(D,\overrightarrow{\bs Z}^{model})=\hat{R}_1(D,(\bs Z^{model}_{i,n-i+1}:1\le i\le n))$ using the chain ladder method with either $\gamma=0$ or $\gamma=1$.\\
Recalling the notation introduced above the ``true'' claims development loss is given by 
$$
\bs X=\bs Z(\theta)+\hat{R}_1(D;\overrightarrow{\bs Z}(\theta))-\hat{R}_0(D)
$$
and the modelled claims development loss is equal to
$$
\bs X^{model}=\bs Z^{model}+\hat{R}_1(D,\overrightarrow{\bs Z}^{model})-\hat{R}_0(D)
$$
where $\bs Z^{model}=\sum_{i=1}^n \bs Z_{i,n-i+1}^{model}$. Let $\bs S^{model}=\bs Z^{model}+\hat{R}_1(D;\overrightarrow{\bs Z}^{model})$.
\begin{rem}
	Since $\SCR(\alpha;D;M)=F_{\bs S^{model}}^{-1}(\alpha)-\hat{R}_0(D)$ and $\bs X=\bs S-\hat{R}_0(D)$, the solvency requirement $\bs X\le SCR(\alpha;\bs D;M)$ is equivalent to $\bs S\le F_{\bs S^{model}}^{-1}(\alpha)$, i.e. \textsl{the best estimate reserve $\hat{R}_0(D)$ cancels out on both sides} of the inequality. This shows that the problem of risk capital calculation for the reserve risk in the context of Solvency II according to Question \ref{qu:2} differs considerably from the objective of classical reserving methods focussing on the estimation error according to the best estimate reserve $\hat{R}_0(D)$.
\end{rem}
\section[The undertaking's perspective]{Parameter uncertainty from the undertaking's perspective}\label{sec:perspective}
Recall the meaning of the terms ``theoretical perspective`'' and ``undertaking's perspective'' from the introduction.\\
The objective of this section is to provide an intuitive understanding of the impact of parameter uncertainty for the reserve risk from an economic point of view. For this purpose it is crucial to recognize that the actual situation of the undertaking can be characterized by the following simple, but fundamental observations:
\begin{enumerate}
	\item[i)] The observed development triangle $D$ is fixed. Hence, for a given estimation method both the parameter estimate $\hat{\theta}$ and the best estimate reserve $\hat{R}_0(D)$ are also fixed. In particular, there is uncertainty about the true parameter vector $\theta$ - not about $\hat{\theta}$.
	\item[ii)] The real economic risk results from the true distribution of future claims depending on the unknown true parameter vector $\theta$. The parameter risk arises from the uncertainty about $\theta$. However, the parameter estimation does not (directly) influence the true distribution of future claims.
\end{enumerate}
These observations make obvious that the actual economic risk of the undertaking relevant for Solvency II is given by the undertaking's perspective.
In particular, in order to model the real economic reserve risk it does not make sense to use a predictive distribution of future claims payments resp. of $\bs X^{model}$ directly based on the distribution of the estimate $\hat{\bs \theta}$. We conclude that the theoretical perspective is not appropriate to model the impact of parameter uncertainty from the economic point of view of the undertaking. More precisely, using the notation introduced in Subsection \ref{subsec:angemessen}:\\
\par
\begingroup
\leftskip4em
\rightskip\leftskip
\textsl{A predictive distribution modelling future losses $\bs X^{model}$ by using the distribution of the estimate $\hat{\bs \theta}$ as the parameter distribution $\mathcal{P}$ of the modelled parameter vector $\bs \theta_{sim}$ is, in general, not appropriate to represent the impact of parameter uncertainty upon the real risk of the undertaking, which arises from the uncertainty about the true parameter vector $\theta$ corresponding to the actual economic losses $\bs X=\bs X(\theta)$.}
\\
\par
\endgroup
Summarizing the discussion above, from the economic point of view of the undertaking parameter risk is defined as follows: 
\begin{definition}\label{def:paramrisk}
	The \textbf{parameter risk} from the undertaking's perspective refers to the uncertainty about the true parameter vector $\theta$ corresponding to the random variable $\bs S$ conditioned on the fixed observed triangle $D$.
\end{definition}
This leads to the question: How can we model parameter risk from the undertaking's perspective? Thus, using the notation introduced in Section \ref{subsec:angemessen}, our objective is to deduce a parameter distribution to model the parameter vector $\bs \theta_{sim}$ reflecting the uncertainty of the undertaking about the true parameter vector\footnote{Recall the difference between modelled and true quantities from Remark \ref{rem:parameter}. In particular, it is important  to understand that we do not assume the true unknown parameter to be random.} $\theta$ based on the fixed observation $\hat{\theta}$. 

There may be several possible methods to model parameter risk from the undertaking's perspective including the Bayesian approach. However, a straight forward idea making no use of any assumptions about an a-priori parameter distribution is based on Fisher's fiducial argument\footnote{For an introduction to fiducial inference and for a discussion of the history as well as the strengthens and weaknesses of the original fiducial approach we refer to \cite{froehlich2014}, \cite{zabell92}. The fiducial approach has properly been generalized in the last decades by many researches (see e.g. \cite{hannig2016} for a comprehensive survey).} (cf. \cite{fisher1930}). Indeed, in \cite{froehlich2014} the authors introduced a method to model parameter uncertainty based on the fiducial approach solving our central Question \ref{qu:2} for a wide class of distributions.\\

Before considering the rather complex reserve risk we illustrate the difference between the theoretical perspective and the modelling of the undertaking's perspective based on fiducial inference using the simple example of the normal distribution $N(0;\sigma^2)$ with fixed, but unknown parameter $\sigma^2$.
In this simple example we give a short proof that the fiducial parameter distribution modelling the uncertainty about $\sigma^2$ from the undertaking's perspective actually yields an exact solution to our central Question \ref{qu:2}. This illustrates the close connection between our central Question \ref{qu:2} and the fundamental idea of viewing parameter uncertainty from the undertaking's perspective. 
\begin{example}\label{ex:intro}
	Let $\bs D=(\bs X_1,\ldots,\bs X_n)$ be the historical data where $\bs X_1,\ldots,\bs X_n$ are independent copies of $\bs X=\sigma\cdot \bs Z$, $\bs Z\sim N(0;1)$ for $i=1,\ldots,n$. An unbiased estimate of the parameter $\sigma^2$ is given by
	\begin{equation}\label{eq:fisher}
	\hat{\bs \sigma}^2 :=\frac{1}{n}\cdot \sum_i \bs X_i^2
	= \frac{\sigma^2}{n}\cdot \sum_i \bs Z_i^2=\frac{\sigma^2}{n}\cdot \bs M
	\end{equation}
	where $\bs Z_i$ are i.i.d. with $\bs Z_i\sim N(0;1)$ and $\bs M:=\sum \bs Z_i^2$ is $\chi^2(n)$ distributed with $n$ degrees of freedom. In this example we consider the modelling of the parameter uncertainty from the theoretical perspective resp. from the undertaking's perspective and denote the modelled parameter by $\bs \sigma_{sim}^2$.
	\begin{enumerate}
		\item 
		From the \textsl{theoretical perspective} $\hat{\bs \sigma}^2$ has the distribution
		\begin{align*}
		\hat{\bs \sigma}^2= \frac{\sigma^2}{n}\cdot \bs M\sim \frac{\sigma^2 }{n} \cdot \chi^2(n) \qquad (A),
		\end{align*}
		where $\chi^2(n)$ is the $\chi^2$-distribution with $n$ degrees of freedom. 
		However, using the distribution $(A)$ to model $\bs \sigma_{sim}^2$ would not reflect the parameter risk from the undertaking's perspective (see the general arguments above). 
		\item From the \textsl{undertaking's perspective} the estimate $\hat{\sigma}^2$ is given, but there is uncertainty about the true parameter $\sigma^2$. Note that from this perspective the uncertainty about $\sigma^2$ is due to the fact that the undertaking does not know the realization of the random factor $\bs M$, since the undertaking would be able to conclude the true value of $\sigma^2$ from $(A)$ if it knew the realization of $\bs M$. In this case there would be no parameter uncertainty. However, since the undertaking has no information about the realization of $\bs M$, the idea of the fiducial approach is to model this uncertainty by using an independent copy\footnote{Note that the assumption that $\bs M$ and $\bs M^\prime$ are independent corresponds to the fundamental idea of any Monte-Carlo based risk model to simulate independent copies of the true risk factors, while the realizations of the true risk factors are unknown. Thus, in general, modelled risk factors are independent of the (unknown) realizations of the true risk factors.}	$\bs M^\prime$ of $\bs M$ (cf. e.g. \cite{hannig2016}, p. 6). Thus, following the fiducial approach we solve Equation $(A)$ for $\sigma^2$ to obtain the modelled parameter  
		$$
		\bs \sigma_{sim}^2=n\cdot \hat{\sigma}^2/\bs M^\prime\sim n\cdot \hat{\sigma}^2/\chi^2(n)  \qquad (B).
		$$
		Note that unlike the Bayesian approach we did neither need any a-priori distributional assumptions to deduce the fiducial distribution nor we assume $\sigma^2$ to be random (since we carefully distinguish between the modelled parameter $\bs\sigma_{sim}^2$ and the unknown true parameter $\sigma^2$).\\
		Indeed, the parameter distribution (B) yields an exact solution to Question \ref{qu:2} for this simple example: Let $\bs X^{model}=\bs \sigma_{sim}\cdot \bs Z^\prime=\hat{\sigma}\cdot\sqrt{\frac{n}{\bs M^\prime}}\cdot\bs Z^\prime$ for some standard normally distributed random variable $\bs Z^\prime$ independent of $\bs Z$, $\bs M$ and $\bs M^\prime$. Note that $\bs X^{model}$ can be written as $\hat{\sigma}\cdot \bs T^\prime$ where $\bs T^\prime=\sqrt{n/\bs M^\prime}\cdot \bs Z^\prime$ is $t$-distributed with $n$ degrees of freedom. Moreover, we define the $t$-distributed random variable $\bs T=\sqrt{n/\bs M}\cdot \bs Z$. For $SCR(\alpha;D;\mbox{fiducial}):=F^{-1}_{\bs X^{model}}(\alpha)=\hat{\sigma}\cdot F^{-1}_{\bs T^\prime}(\alpha)$ we derive
		\begin{align*}
		P(\bs X\le SCR(\alpha;\bs D;\mbox{fiducial}))&=P\left(\sigma \cdot\bs Z\le \hat{\bs \sigma}\cdot F^{-1}_{\bs T^\prime}(\alpha)\right)\\&=P\left(\sigma \cdot\bs Z\le \sigma\cdot  \sqrt{\frac{\bs M}{n}}\cdot F^{-1}_{\bs T^\prime}(\alpha)\right)\\
		&=P\left(\bs T
		\le F^{-1}_{\bs T^\prime}(\alpha)\right)=\alpha.
		\end{align*}
		This proves that the modelled risk capital $SCR(\alpha; D;\mbox{fiducial})=F^{-1}_{\bs X^{model}}(\alpha)$ attains the required probability of solvency, i.e. we solved Question \ref{qu:2} for this example.
	\end{enumerate} 
	The distributions (A) and (B) do not coincide. Note that the density function of the $\chi^2(n)$-distribution corresponding to $(A)$ is equal to
	$$
	f(x)=const\cdot x^{\frac{n}{2}-1}\exp\left(-\frac{x}{2}\right),
	$$
	i.e. it decreases exponentially, whereas the distribution $1/\chi^2(n)$ corresponding to (B) has the density function 
	$$
	\frac{1}{x^2}\cdot f\left(\frac{1}{x}\right)=const\cdot x^{-\frac{n}{2}-1}\cdot \exp\left(-\frac{1}{2x}\right).
	$$
	Hence, distribution (B) representing the undertaking's perspective has an heavy tail (no exponential decay) in contrast to the distribution (A). 
\end{example}
The example demonstrates that the two perspectives are not equivalent. There is no symmetry or equivalence between the uncertainty about $\hat{\theta}$ from the theoretical perspective and the uncertainty about $\theta$ from the undertaking's perspective. In particular, in this example the parameter risk from the economical relevant perspective from the undertaking would be significantly underestimated using the parameter distribution (A) corresponding to the theoretical perspective.
\begin{rem}
	\begin{enumerate}
\item 	The considerations in Example \ref{ex:intro} can be generalized to a wider class of distribution families (see \cite{froehlich2014}).
\item For the normal distribution used in Example \ref{ex:intro} the fiducial distribution $(B)$ coincides with the Bayesian posterior distribution with non-informative prior (cf. \cite{hora}). This indicates that in Example \ref{ex:intro} the parameter distribution $(B)$ directly corresponds to the undertaking's perspective in the absence of a-priori information.
\item We point out that the distribution of $\bs \sigma_{sim}^2$ is presumed to represent the uncertainty of the undertaking about the true parameter $\sigma^2$ rather than to be a point our interval estimator for $\sigma^2$ (cf. \cite{hannig2016}, p.6). 
\end{enumerate}
\end{rem}
In this contribution we restrict to the undertaking's perspective which is based on the observed claims development triangle $D$. Note that the bootstrapping approach (see e.g. \cite{pinheiro2003,england2006,bjoerkwall2011}) as well as many contributions concerning the mean squared error of estimation (see e.g. \cite{Mack1993,Mack1999,merz2008}) adopt the theoretical perspective. Thus, these approaches are not appropriate to model the impact of parameter uncertainty upon the actual risk of economic losses from the undertaking's perspective.\\ 
Indeed, in Section \ref{subsec:boots} we demonstrate that bootstrapping does not yield a solution to Question \ref{qu:2}.
\section{Appropriateness of the methods for the normal model}\label{sec:exmethods}
\sectionmark{Appropriateness for the normal mode}
Referring to Question \ref{qu:2} in this section we assess several methods for calculating the risk capital for the one-year reserve risk by comparing the probability of solvency attained by the respective method to the required confidence level. To avoid technical complications we concentrate on a rather simple model - the normal model \cite{england2006, gisler2006} based on Mack's chain ladder model (cf. Subsection \ref{subsec:grand}).\\
We assume that the individual chain ladder factors $\bs F_{i,k}=\bs C_{i,k}/C_{i,k-1}$ conditioned on $\{C_{i,0},\ldots,C_{i,k-1}\}$ are normally distributed, i.e. there exists parameters $f_k>0$ and $\sigma_k$ independent of the specific accident year $i$ such that
$$
\bs F_{i,k}|\{C_{i,0},\ldots,C_{i,k-1}\}\sim N\left(f_k;\frac{\sigma_k^2}{C_{i,k-1}^\gamma}\right) 
$$
for $0\le i\le n$ and $1\le k\le n$ and $\gamma=0$ or $\gamma=1$.\\
Thus, we can write $\bs F_{i,k}$ in the form (\ref{eq:res}) with independent, standard normally distributed residues $\bs \zeta_{i,k}$. We denote the set of realized (``true'', but still unknown) residues by 
$\mathcal{R}=\{\zeta_{i,k}:i+k\le n\}$.\\ 
Moreover, we assume that $\bs C_{i,0}$ is normally distributed with mean $f_0$ and variance $\sigma_0^2$ for $0\le i\le n$. Note that the realizations of $C_{i,0}$ have already been observed. Hence, for the modelling of the reserve risk an estimation of the parameters $f_0$ and $\sigma_0^2$ is not necessary and the parameter uncertainty with respect to these parameters is not relevant. \\ 


\subsection{Without modelling parameter risk}\label{subsec:without}
In this section we consider the approach without modelling parameter risk, i.e. we set $\bs f_k^{sim}\equiv \hat{f}_k$ and $\bs \sigma_k^{sim}\equiv \hat{\sigma}_k$ (cf. Subsection \ref{subsec:angemessen}). Note that the common approach used in practice is bootstrapping which will be considered in Subsection \ref{subsec:boots}.\\
We model the cumulative claims for the next business year by
$
\bs C_{i,n-i+1}^{model,without}=C_{i,n-i}\cdot \bs F_{i,n-i+1}^{sim,without}
$
with 
$$
\bs F_{i,n-i+1}^{sim,without}=\hat{f}_{n-i+1}+\frac{\hat{\sigma}_{n-i+1}}{\sqrt{C_{i,n-i}^\gamma}}\cdot \bs \zeta_{i,n-i+1}^\prime
$$
for $\bs \zeta_{i,n-i+1}^\prime$ independent, normally distributed residues.\\
We set $$\bs Z^{model}_{without}(D)=\sum_{i=1}^n \bs Z_{i,n-i+1}^{model,without}\bs =\sum_{i=1}^n \left(\bs C_{i,n-i+1}^{model,without}-C_{i,n-i}\right)$$ and obtain the reserve $\hat{R}_1(D,\overrightarrow{\bs Z}^{model}_{without})=\hat{R}_1(D,(\bs Z_{i,n-i+1}^{model,without}:1\le i\le n))$ in $t=1$. The risk capital $\SCR(\alpha;D;without)$ is defined as the $\alpha$-quantile of 
$$
\bs X^{model}_{without}(D) =\bs Z^{model}_{without}(D) + \hat{R}_1(D,\overrightarrow{\bs Z}^{model}_{without})-\hat{R}_0(D).
$$
The estimates $\hat{f}_k$ and $\hat{\sigma}_k^2$ depend on the realization $D$ of the random claims development triangle $\bs D$.\\
Using the assumptions given in Subsection \ref{subsec:data} and following the general approach described in Subsection \ref{subsec:general} we derive the results for the probability of solvency presented in Table \ref{tab:without}. 
\begin{table}[H]
\begin{center}
\begin{tabular}{|c|c|c|}
\hline
\mbox{$\alpha$} &$\gamma=0$ &$\gamma=1$\\
\hline
\hline
90\% &84.98\%&85.80\%\\
\hline
95\% &91.19\%&91.35\%\\
\hline
99\% &96.97\%&97.07\%\\
\hline
99.5\% &97.97\%&98.09\%\\
\hline
\end{tabular}
\caption{Solvency probabilities $P(\bs X\le \SCR(\alpha;\bs D;without))$ for the approach without the consideration of parameter risk, for different quantiles and for $\gamma=0$ resp. $\gamma=1$}\label{tab:without}
\end{center}
\end{table}
\begin{conclusion}
Neglecting parameter uncertainty leads to a probability of solvency which is significantly lower than the required confidence level.
\end{conclusion}
\subsection{Bootstrapping}\label{subsec:boots}
In this section we consider the popular approach using bootstrapping. There are numerous variants of the bootstrapping approach; we follow Subsection 7.4 in \cite{wuethrich2008}.\\
In the sequel we describe how to determine the modelled risk $\bs X_{model}^{BT}(D)$ using bootstrapping: Again, we set $\hat{f}_k$ und $\hat{\sigma}_k^2$ as in (\ref{eq:estim}).
Given $\hat{f}_k$ und $\hat{\sigma}_k^2$ we estimate the residues by 
$$
\hat{\zeta}_{i,k}=\left(\frac{F_{i,k}-\hat{f}_k}{\hat{\sigma}_k}\right)\cdot \sqrt{C_{i,k-1}^\gamma} \mbox{ for $i+k\le n$, $k<n$}
$$
and consider the set $\hat{\mathcal{R}}:=\{\hat{\zeta}_{i,k}:i+k\le n, k<n\}.$\\
As pointed out in \cite{wuethrich2008}, Section 7.4, Equation 7.23, the variance of the residues $\hat{\bs \zeta}_{\cdot,k}\in\mathcal{R}$ is smaller than $1$. More precisely,
$$
\Var(\hat{\bs \zeta}_{i,k}|C_{0,k-1},\ldots,C_{n-k,k-1})=1-\frac{C_{i,k-1}^\gamma}{\sum_{i=0}^{n-k} C_{i,k-1}^\gamma}<1.
$$
We adjust the residues accordingly (cf. Equation 7.24 in \cite{wuethrich2008}) and obtain the set $\mathcal{R}^*$.\\
We follow the conditional approach in \cite{wuethrich2008}, Section 7.4.2 to construct a bootstrapping distribution of $\bs \theta^{sim}$ by Monte-Carlo simulation.\\ 
A scenario of the bootstrapping distribution $(\bs f_k^{sim,BT}, (\bs \sigma_k^{sim,BT})^2)$ is constructed as follows: We determine the chain ladder factors $\bs f_k^{sim,BT}$ by
$$
\bs f_k^{sim,BT}=\sum_{i=0}^{n-k}\frac{C_{i,k-1}^\gamma}{\sum_{h=0}^{n-k} C_{h,k-1}^\gamma}\cdot \bs F_{i,k}^{*,BT}
$$
with
$$
\bs F_{i,k}^{*,BT}=\hat{f}_k+\frac{\hat{\sigma}_k}{\sqrt{C_{i,k-1}^\gamma}}\cdot \bs \zeta_{i,k}^*
$$
where $\bs \zeta_{i,k}^*$ is chosen randomly from $\mathcal{R}^*$ and set
$$
\left(\bs \sigma_k^{sim,BT}\right)^2=\frac{1}{n-k}\cdot \sum_{i=0}^{n-k} C_{i,k-1}^\gamma\cdot (\bs F_{i,k}^{*,BT}-\bs f_k^{sim,BT})^2.
$$
We then define
$$
\bs F_{1,n}^{sim,BT}\equiv 1\mbox{ and }
\bs F_{i,n-i+1}^{sim,BT}=\bs f_{n-i+1}^{sim,BT}+\frac{\bs \sigma_{n-i+1}^{sim,BT}}{\sqrt{C_{i,n-i}^\gamma}}\cdot \bs \zeta_{i,n-i+1}^\prime \mbox{ for }i=2,\ldots,n
$$
where $\bs \zeta_{i,n-i+1}^\prime\sim N(0;1)$ are i.i.d. random variables independent of the bootstrapped parameters.
This defines
$
\bs C^{model,BT}_{i,n-i+1}=\bs F^{sim,BT}_{i,n-i+1}\cdot C_{i,n-i}
$
and
$$
\bs Z^{model}_{BT}(D)=\sum_{i=1}^n \bs Z_{i,n-i+1}^{model,BT}=\sum_{i=1}^n \left(\bs C_{i,n-i+1}^{mod,BT}-C_{i,n-i}\right).
$$
The risk capital $\SCR(\alpha;D;BT)$ is defined as the $\alpha$-quantile of 
$$
\bs X^{model}_{BT}(D) =\bs Z^{model}_{BT}(D) + \hat{R}_1(D,\overrightarrow{\bs Z}^{model}_{BT})-\hat{R}_0(D)
$$
where $\overrightarrow{\bs Z}^{model}_{BT}=(\bs Z_{i,n-i+1}^{model,BT}:1\le i\le n)$.\\
Using the assumptions given in Subsection \ref{subsec:data} and following the general approach described in Appendix \ref{subsec:general} we derive the results for the probability of solvency presented in Table \ref{tab:bt}. 
\begin{table}[H]
\begin{center}
\begin{tabular}{|c|c|c|}
\hline
\mbox{$\alpha$} &$\gamma=0$ &$\gamma=1$\\
\hline
\hline
90\% &88.57\%&89.05\%\\
\hline
95\% &93.68\%&94.02\%\\
\hline
99\% &98.27\%&98.63\%\\
\hline
99.5\% &99.08\%&99.15\%\\
\hline
\end{tabular}
\caption{Solvency probabilities $P(\bs X\le \SCR(\alpha;\bs D;BT))$ for the bootstrapping approach, for different quantiles and for $\gamma=0$ and $\gamma=1$} \label{tab:bt}
\end{center}
\end{table}
\begin{conclusion}
The bootstrapping approach does not attain the required confidence level.
\end{conclusion}
\subsection{The inversion method for the normal model}\label{subsec:inv}
The ``inversion method'' introduced in \cite{froehlich2014} expresses the parameter estimate $\hat{\theta}$ in terms of the true parameter $\theta$ and to invert this relation to obtain an expression of $\theta$ in terms of $\hat{\theta}$.\\ 
Since we assume that all payments are settled after $n$ development years, we set $\bs f_n^{sim}\equiv 1$ and $\bs \sigma_n^{sim}\equiv 0$ and apply the idea of the inversion method to the parameter vector $\theta=(f_1,\sigma_1^2,f_2,\sigma_2,\ldots, f_{n-1},\sigma_{n-1}^2)$ of the normal model: Inserting
$$
\bs F_{i,k}=f_k+\frac{\sigma_k}{\sqrt{C_{i,k-1}^\gamma}}\cdot \bs \zeta_{i,k}
$$
with $\bs \zeta_{i,k}$ independent, standard normally distributed into (\ref{eq:estim}) yields for $k=1,\ldots,n-1$  
\begin{align*}
\hat{\bs f}_k&=\sum_{i=0}^{n-k} \bs F_{i,k}\cdot \frac{C_{i,k-1}^\gamma}{\sum_{h=0}^{n-k} C_{h,k-1}^\gamma}\\
&=f_k+\sigma_k\cdot \sum_{i=0}^{n-k} \frac{\bs \zeta_{i,k}}{\sqrt{C_{i,k-1}^\gamma}}\cdot \frac{C_{i,k-1}^\gamma}{\sum_{h=0}^{n-k} C_{h,k-1}^\gamma}=f_k+\sigma_k\cdot \bs R_{k} \mbox{ and }\\
\hat{\bs \sigma}_k^2&=\frac{1}{n-k}\sum_{i=0}^{n-k} C_{i,k-1}^\gamma (\bs F_{i,k}-\hat{\bs f}_k)^2\\
&=\frac{1}{n-k}\cdot \sigma_k^2\cdot \sum_{i=0}^{n-k} C_{i,k-1}^\gamma\left(\frac{\bs \zeta_{i,k}}{\sqrt{C_{i,k-1}^\gamma}}-\sum_{j=0}^{n-k}\frac{\bs \zeta_{j,k}}{\sqrt{C_{j,k-1}^\gamma}}\frac{C_{j,k-1}^\gamma}{\sum_{h=0}^{n-k} C_{h,k-1}^\gamma}\right)^2\\
&=\sigma_k^2\cdot \bs M_{k}
\end{align*}
with \label{def:rk}
\begin{equation}
\label{eq:rk}
\bs R_{k}=\sum_{i=0}^{n-k} \bs \zeta_{i,k}\cdot \frac{\sqrt{C_{i,k-1}^\gamma}}{\sum_{h=0}^{n-k} C_{h,k-1}^\gamma}
\end{equation}
and
\begin{eqnarray}\label{eq:tk}
\bs M_{k}&=\frac{1}{n-k}\sum_{i=0}^{n-k} C_{i,k-1}^\gamma\left(\frac{\bs \zeta_{i,k}}{\sqrt{C_{i,k-1}^\gamma}}-\sum_{j=0}^{n-k}\frac{\bs \zeta_{j,k}}{\sqrt{C_{j,k-1}^\gamma}}\frac{C_{j,k-1}^\gamma}{\sum_{h=0}^{n-k-1} C_{h,k-1}^\gamma}\right)^2 \nonumber\\ 
&=\sum_{i=0}^{n-k} C_{i,k-1}^\gamma\left(\frac{\bs \zeta_{i,k}}{\sqrt{C_{i,k-1}^\gamma}}-\bs R_k\right)^2.
\end{eqnarray}
Solving these equations for $f_k$ resp. $\sigma_k^2$ for $1\le k\le n-1$ defines a probability distribution of the unknown parameter vector $(f_k,\sigma_k^2)$ given by
\begin{align}
\bs f_k^{sim}:=\hat{f}_k-\frac{\hat{\sigma}_k}{\sqrt{\bs M_{k}^\prime}}\cdot \bs R_{k}^\prime\mbox{\quad and\quad }
\left(\bs \sigma_k^{sim}\right)^2:=\frac{\hat{\sigma}_k^2}{\bs M_{k}^\prime} \label{glAB}
\end{align}
where the modelled random variables $\bs R_k^\prime$ resp. $\bs M_{k}^\prime$ are independent copies of $\bs R_k$ and $\bs M_{k}$ obtained by replacing $\bs \zeta_{i,k}$ in (\ref{eq:rk}) and (\ref{eq:tk}) by independent $\bs \zeta^\prime_{i,k}\sim N(0;1)$.\\
To model the claims development loss of the next calendar year we consider (\ref{glAB}) for $k=n-i+1$ and $i=2,\ldots,n$:
\begin{align*}
\bs f_{n-i+1}^{sim}&=\hat{f}_{n-i+1}-\bs \sigma_{n-i+1}^{sim} \bs R_{n-i+1}^\prime \mbox{ and } \left(\bs \sigma_{n-i+1}^{sim}\right)^2=\frac{\hat{\sigma}_{n-i+1}^2}{\bs M_{n-i+1}^\prime}.
\end{align*}
Let
\begin{equation}\label{eq:fmodell}
\bs F_{i,n-i+1}^{sim}:=\bs f_{n-i+1}^{sim}+\frac{\bs \sigma_{n-i+1}^{sim}}{\sqrt{C_{i,n-i}^{\gamma}}}\cdot \bs \zeta_{i,n-i+1}^\prime
\end{equation}
and $\bs Z_{i,n-i+1}^{model}:=(\bs F_{n-i+1}^{sim}-1)\cdot C_{i,n-i}$ for $i=1,\ldots,n$. Note that $\bs Z_{i,n-i+1}^{model}$ depends on the observed  triangle $D$.\\
However, modelling the claims development result directly by $\bs X^{model}=\sum \bs Z_{i,n-i+1}^{model} + \hat{R}_1(D,\overrightarrow{\bs Z}^{model})-\hat{R}_0(D)$ with $\overrightarrow{\bs Z}^{model}=(\bs Z^{model}_{i,n-i+1}: 1\le i\le n)$ yields a risk capital model which is too conservative, i.e. setting
$
SCR(\alpha;D):=F_{\bs X^{model}}^{-1}(\alpha)
$
yields 
$
P(\bs X^{model}\le SCR(\alpha;\bs D))>\alpha
$
for e.g. $\alpha=99.5\%$  (see \cite{froehlich2015} for a comprehensive discussion).\\
Therefore, we need to adjust the inversion method to derive an risk capital model leading to a significantly better approximation of the desired probability of solvency.\\
For the adjustment of the inversion method we introduce a stochastic correction factor of the same form as suggested in \cite{froehlich2015}:
\begin{equation}\label{eq:asim}
\bs a_{sim}:=\left(\sum_{i=2}^n \frac{\hat{w}_{n-i+1}}{\bs M_{n-i+1}^\prime}\cdot \sum_{i=2}^n w_{n-i+1}\cdot \bs M_{n-i+1}^\prime \right)^{-\frac{1}{2}}
\end{equation}
with weights $w_{n-i+1}$, $2\le i\le n$, defined by
$$
w_{n-i+1}=\frac{(\sigma_{n-i+1})^2\cdot C_{i,n-i}^2\cdot \Bigl(\frac{1}{C_{i,n-i}^\gamma}+\frac{1}{\sum\limits_{l=0}^{i-1} C_{l,n-i}^\gamma}\Bigr)}{\sum_{j=2}^n (\sigma_{n-j+1})^2\cdot C_{j,n-j}^2\cdot \Bigl( \frac{1}{C_{j,n-j}^\gamma}+\frac{1}{\sum\limits_{l=0}^{j-1} C_{l,n-j}^\gamma}\Bigr)}
$$
and
\begin{equation}\label{eq:what}
\hat{w}_{n-i+1}=\frac{(\hat{\sigma}_{n-i+1})^2\cdot C_{i,n-i}^2\cdot \Bigl(\frac{1}{C_{i,n-i}^\gamma}+\frac{1}{\sum\limits_{l=0}^{i-1} C_{l,n-i}^\gamma}\Bigr)}{\sum_{j=2}^n (\hat{\sigma}_{n-j+1})^2\cdot C_{j,n-j}^2\cdot \Bigl( \frac{1}{C_{j,n-j}^\gamma}+\frac{1}{\sum\limits_{l=0}^{j-1} C_{l,n-j}^\gamma}\Bigr)}.
\end{equation}
Set
$
\hat{Z}_{i,n-i+1}=(\hat{f}_{n-i+1}-1) \cdot C_{i,n-i}
$
and
$$
\bs Z^{model,adj}_{i,n-i+1}=(1-\bs a_{sim})\hat{Z}_{i,n-i+1}+ \bs a_{sim}\cdot \bs Z_{i,n-i+1}^{model}
$$
for $i=1,\ldots,n.$\\
Let $\bs Z^{model}_{adj}:=\sum_{i=1}^n \bs Z_{i,n-i+1}^{model,adj}$ and let $\mathcal{R}=\{\zeta_{i,k}:i+k\le n\}$ be the set of realized residues. To express the dependency of $\bs Z^{model}_{adj}$ on the residues $\mathcal{R}$ we write $\bs Z^{model}_{adj}= \bs Z^{model}_{adj}(\mathcal{R})$. 
The following theorem holds independently of the particular choice of the parameter vector $\theta$.
\begin{theorem}\label{thm:main}
	Let $\SCR_{\bs Z}(\alpha;\mathcal{R};M)$ be the $\alpha$-quantile of $\bs Z^{model}_{adj}=\bs Z^{model}_{adj}(\mathcal{R})$. Then
$$
P\left(\sum_{i=1}^n \bs Z_{i,n-i+1}\le \SCR_{\bs Z}(\alpha;\boldsymbol{\mathcal{R}};M)\right)=\alpha.
$$
\end{theorem}
\begin{proof}
	See Appendix \ref{sec:app}.
\end{proof}
Note that, since $\sigma_{n-i+1}$ is unknown, we can only estimate $\bs a_{sim}$. In the sequel we use the following estimate 
$$
\hat{\bs a}_{sim}:=\left(\sum_{i=2}^{n} \frac{\hat{w}_{n-i+1}}{\bs M_{n-i+1}^\prime}\cdot \sum_{i=2}^{n} \hat{w}_{n-i+1}\cdot \bs M_{n-i+1}^\prime \right)^{-\frac{1}{2}}.
$$
Theorem \ref{thm:main} motivates to set
$$
\hat{\bs Z}_{i,n-i+1}^{model,adj}:=(1-\hat{\bs a}_{sim})\cdot \hat{Z}_{i,n-i+1} + \hat{\bs a}_{sim}\cdot \bs Z_{i,n-i+1}^{model}
$$
and
$
\hat{\bs Z}^{model}_{adj}:=\sum_{i=1}^n \hat{\bs Z}_{i,n-i+1}^{model,adj}.
$\\
Define the modelled risk by
$$
\hat{\bs X}^{model}_{adj}(D)=\hat{\bs Z}^{model}_{adj}(D)+\hat{R}_1(D,(\hat{\bs Z}_{i,n-i+1}^{model,adj}:1\le i\le n))-\hat{R}_0(D)
$$ 
and model the risk capital $\SCR(\alpha;D;model,adj)$ as the $\alpha$-quantile of $\hat{\bs X}^{model}_{adj}(D)$.\\
Again we consider the example given in Subsection \ref{subsec:data} and obtain the following probabilities of solvency:
\begin{table}[H]
\begin{center}
\begin{tabular}{|c|c|c|}
\hline
\mbox{$\alpha$} &$\gamma=0$ &$\gamma=1$\\
\hline
\hline
90\% &89.92\%&89.76\%\\
\hline
95\% &95.06\%&94.89\%\\
\hline
99\% &99.03\%&98.94\%\\
\hline
99.5\% &99.51\%&99.48\%\\
\hline
\end{tabular}
\end{center}
\caption{Solvency probabilities $P(\bs X\le \SCR(\alpha;\bs D;model,adj))$ for the modified inversion method for different quantiles and for $\gamma=0$ and $\gamma=1$} 
\label{tab:inv}
\end{table}
\begin{rem}
	Note that the probability of solvency defined in (\ref{eq:konkret}) depends on the random triangle $\bs D$. In Theorem \ref{thm:main} we only consider the probability of solvency depending on the randomness of the residues $\boldsymbol{\mathcal{R}}$ for fixed weights $C_{i,k}$. Moreover, the theorem focuses on the $\alpha$-quantile of the random variable $\bs Z$ rather than the complete development loss given by $\bs X = \hat{\bs R}_1 + \bs Z - \hat{R}_0$. However, the experimental results in Table \ref{tab:inv} demonstrate that the method also works for both a random $\bs D$ and the complete development loss $\bs X$.
\end{rem}
\begin{conclusion}
The risk capital model based on the adjustment of the inversion method using the stochastic correction factor $\hat{\bs a}_{sim}$ yields a probability of solvency very close to the required confidence levels, i.e. it provides an answer to Question \ref{qu:2} posed in the introduction in good approximation.
\end{conclusion}
\subsection{Effect on the risk capital}
We consider the effect on the risk capital calculation for an explicit example.  Consider the
claims development triangle taken from \cite{merz2008}:

	\begin{table}[H]
		{\scriptsize
		\begin{tabular}{|c|c|c|c|c|c|c|c|c|c|}
			\hline
			&\multicolumn{9}{|c|}{\scriptsize Development year}\\
			\hline
			&\scriptsize 0 &\scriptsize1 &\scriptsize2 &\scriptsize3 &\scriptsize4 &\scriptsize5 &\scriptsize6 &\scriptsize7 &\scriptsize8 \\
			\hline
			\scriptsize0 &\scriptsize2,202,584	&\scriptsize3,210,449	&\scriptsize3,468,122	&\scriptsize3,545,070	&\scriptsize3,621,627	&\scriptsize3,644,636&\scriptsize3,669,012	&\scriptsize3,674,511	&\scriptsize3,678,633\\
			\scriptsize	1 &\scriptsize2,350,650	&\scriptsize3,553,023	&\scriptsize3,783846	&\scriptsize3,840067	&\scriptsize3,865,187	&\scriptsize3,878,744 &\scriptsize3,898,281	&\scriptsize3,902,425 &	\\
			\scriptsize	2 &\scriptsize2,321,885	&\scriptsize3,424,190	&\scriptsize3,700,876	&\scriptsize3,798198	&\scriptsize3,854,755	&\scriptsize3,878,993&\scriptsize3,898,825&&	\\
			\scriptsize3 &\scriptsize2,171,487	&\scriptsize3,165,274	&\scriptsize3,395,841	&\scriptsize3,466453	&\scriptsize3,515,703	&\scriptsize3,548,422&&&	\\		
			\scriptsize4 &\scriptsize2,140,328	&\scriptsize3,157,079	&\scriptsize3,399,262	&\scriptsize3,500,520	&\scriptsize3,585,812&\scriptsize&\scriptsize&\scriptsize&\scriptsize	\\			
			\scriptsize5 &\scriptsize2,290,664	&\scriptsize3,338,197	&\scriptsize3,550,332	&\scriptsize3,641,036	&\scriptsize&\scriptsize&\scriptsize&\scriptsize&\scriptsize	\\			
			\scriptsize6 &\scriptsize2,148,216	&\scriptsize3,219,775	&\scriptsize3,428,335&\scriptsize&\scriptsize&\scriptsize	&\scriptsize&\scriptsize&\scriptsize\\					
			\scriptsize7 &\scriptsize2,143,728	&\scriptsize3,158,581	&\scriptsize&\scriptsize&\scriptsize&\scriptsize&\scriptsize&\scriptsize&\scriptsize	\\					
			\scriptsize8 &\scriptsize2,144,738	&&&&&&&&	\\
			\hline						
		\end{tabular}
		%
	}
	\end{table}

The chain ladder reserve $\hat{R}_0(D)$ for $\gamma=0$ equals to 2,243,574 Euro and for $\gamma=1$ equals to 2,237,826 Euro.\\
The risk capital calculation yields the following results for the modelled risk capital with respect of the $99.5\%$-quantile using the approaches discussed in the previous sections:
\begin{table}[H]
	\begin{center}
		\begin{tabular}{|c|c|c|c|}
			\hline
			&without  &with &with adjusted\\
			 &parameter uncertainty &bootstrapping  &inversion method\\
			\hline
			\hline
			$\gamma=0$ &191,589 &216,115	&227,182\\
			\hline
			$\gamma=1$ &194,916 &216,365	&226,980\\
			\hline	
		\end{tabular}
	\end{center}
\end{table}
\section{Summary and Outlook} 
\label{sec:outlook}
This article deals with the internal modelling of parameter uncertainty for the reserve risk. We pointed out that for the consideration of parameter uncertainty, the undertaking's perspective is the adequate perspective referring to the real risk of economic losses. Therefore, in order to model parameter uncertainty for the reserve risk in the context of Solvency II it is not appropriate to apply methods of classical reserving designed to measure the prediction error from the theoretical perspective.\\
Considering the probability of solvency already introduced in \cite{gerrard,froehlich2014} we assessed several methods to model parameter uncertainty for risk capital calculations considering a very simple model - the normal model. In particular, we demonstrate that the popular bootstrapping approach does not guarantee the required probability of solvency. We then presented an adjustment of the inversion method introduced in \cite{froehlich2014} achieving the required probability of solvency in good approximation.\\
\\
The main message of our article is not to recommend the usage of the normal model (together with the inversion method). Rather we stress the importance of modelling the solvency capital requirement in such a way that it meets the desired confidence level of 99.5\% - even under the consideration of parameter uncertainty. The normal model is just used for illustration.\\
There are still many questions left for future research:
\begin{enumerate}
\item The normal model is very simple and rarely used in practice.  
For other well-established models the question how to guarantee the required probability of solvency is still open.
\item Does there exist a parameter distribution that guarantees the required probability of solvency simultaneously on every aggregation level (i.e. on the level of every single development factor, every single accident year, every line of business as well as on the level of the overall risk) without using any correction factor when proceeding from one aggregation level to another (cf. \cite{froehlich2015})? 
\item Throughout this contribution we assumed that all claims are settled after $n$ years. In particular, we did not address the problem of parameter uncertainty in the context of tail modelling. 
\end{enumerate}
\textbf{Acknowledgments.}
The experimental results presented in Section \ref{sec:exmethods} have been generated using a Java program. We are grateful for the opportunity to run the program on the bwGriD cluster of the Hochschule Esslingen.\\
The work of the second author has been supported by the DVfVW (Deutscher Verein f\"ur Versicherungswissenschaft) by a Modul 1 Forschungsprojekt with the title ``Das Parameterrisiko in Risikokapitalberechnungen für Versicherungsbest\"ande''.

\begin{appendix}
\section{Proof of Theorem 4.3}\label{sec:app}
We prove Theorem \ref{thm:main} stated in Section \ref{subsec:inv}.

\begin{proof} 
Inserting (\ref{eq:what}) and (\ref{glAB}) into the definition of $\bs a_{sim}$ (cf. Equation \ref{eq:asim}) yields
$$
	\bs a_{sim}=\sqrt{
		\frac{\sum\limits_{i=2}^n (\hat{\sigma}_{n-i+1})^2\cdot C_{i,n-i}^2\cdot \left(\frac{1}{C_{i,n-i}^\gamma}+\frac{1}{\sum_{l=0}^{i-1} C_{l,n-i}^\gamma}\right)}
		{\sum\limits_{i=2}^n (\bs \sigma_{n-i+1}^{sim})^2\cdot C_{i,n-i}^2\cdot \left(\frac{1}{C_{i,n-i}^\gamma}+\frac{1}{\sum_{l=0}^{i-1} C_{l,n-i}^\gamma}\right)\sum\limits_{i=2}^n w_{n-i+1}\cdot \bs M_{n-i+1}^\prime}}.
$$
Since we assume that all claims are settled after $n$ development years, $\bs Z_{1,n}=\bs Z_{1,n}^{model}\equiv \hat{Z}_{1,n}=0$ for all $n$.\\ 
We have
\begin{align*}
\bs Z^{model}_{adj}=&(1-\bs a_{sim})\sum\limits_{i=2}^n \hat{Z}_{i,n-i+1}+\bs a_{sim} \sum\limits_{i=2}^n \bs Z_{i,n-i+1}^{model}\\
=&(1-\bs a_{sim})\sum\limits_{i=2}^n (\hat{f}_{n-i+1}-1)\cdot C_{i,n-i}\\
&+\bs a_{sim}\sum\limits_{i=2}^n \left(\hat{f}_{n-i+1}-\bs \sigma_{n-i+1}^{sim}\cdot \bs R_{n-i+1}^\prime+\frac{\bs \sigma_{n-i+1}^{sim}}{\sqrt{C_{i,n-i}^\gamma}}\cdot \bs \zeta_{i,n-i+1}^\prime-1\right)\cdot C_{i,n-i}\\
\sim &\sum\limits_{i=2}^n (\hat{f}_{n-i+1}-1)\cdot C_{i,n-i}+\bs a_{sim}\sqrt{\sum\limits_{i=2}^n C_{i,n-i}^2\left(\frac{1}{C_{i,n-i}^\gamma}+\frac{1}{\sum_{l=0}^{i-1} C_{l,n-i}^\gamma}\right)(\bs \sigma_{n-i+1}^{sim})^2}\cdot\bs\zeta^\prime
\end{align*}
with independent, standard normally distributed random variables $\bs\zeta^\prime$ and $\bs \zeta_{i,n-i+1}^\prime$ independent of $\bs M_{n-i+1}^\prime$.\\
Hence,
\begin{align*}
&\bs Z^{model}_{adj}\sim \sum\limits_{i=2}^n (\hat{f}_{n-i+1}-1)\cdot C_{i,n-i}+
\sqrt{\frac{\sum\limits_{i=2}^n \hat{\sigma}_{n-i+1}^2 C_{i,n-i}^2\left(\frac{1}{C_{i,n-i}^\gamma}+\frac{1}{\sum_{l=0}^{i-1} C_{l,n-i}^\gamma}\right)}{\sum\limits_{i=2}^n w_{n-i+1}\bs M_{n-i+1}^\prime}}\cdot\bs\zeta^\prime\\
&=\sum\limits_{i=2}^n (\hat{f}_{n-i+1}-1)\cdot C_{i,n-i}\\
&+\sqrt{
\frac{\sum\limits_{i=2}^n \hat{\sigma}_{n-i+1}^2 C_{i,n-i}^2\left(\frac{1}{C_{i,n-i}^\gamma}+\frac{1}{\sum_{l=0}^{i-1} {C_{l,n-i}^\gamma}}\right)\sum_i \sigma_{n-i+1}^2 C_{i,n-i}^2\left(\frac{1}{C_{i,n-i}^\gamma}+\frac{1}{\sum_{l=0}^{i-1} C_{l,n-i}^\gamma}\right)}{\sum\limits_{i=2}^n \sigma_{n-i+1}^2 C_{i,n-i}^2 \left(\frac{1}{C_{i,n-i}^\gamma}+\frac{1}{\sum_{l=0}^{i-1} C_{l,n-i}^\gamma}\right)\bs M_{n-i+1}^\prime}}\cdot\bs\zeta^\prime\\
&=\sum\limits_{i=2}^n (f_{n-i+1}-1)\cdot C_{i,n-i}+\sum\limits_{i=2}^n \sigma_{n-i+1}\cdot R_{n-i+1}\cdot C_{i,n-i}\\
&\scriptstyle+\sqrt{\frac{\sum\limits_{i=2}^n \sigma_{n-i+1}^2 C_{i,n-i}^2\left(\frac{1}{C_{i,n-i}^\gamma}+\frac{1}{\sum_{l=0}^{i-1} {C_{l,n-i}^\gamma}}\right) \cdot M_{n-i+1}\cdot\sum\limits_{i=2}^n \sigma_{n-i+1}^2 C_{i,n-i}^2\left(\frac{1}{C_{i,n-i}^\gamma}+\frac{1}{\sum_{l=0}^{i-1} C_{l,n-i}^\gamma}\right)}{\sum\limits_{i=2}^n \sigma_{n-i+1}^2 C_{i,n-i}^2 \left(\frac{1}{C_{i,n-i}^\gamma}+\frac{1}{\sum_{l=0}^{i-1} C_{l,n-i}^\gamma}\right)\bs M_{n-i+1}^\prime}}\cdot\bs\zeta^\prime
\end{align*}
where $R_{n-i+1}$ is a realization of $\bs R_{n-i+1}$ such that $\hat{f}_{n-i+1}=f_{n-i+1}+\sigma_{n-i+1}\cdot R_{n-i+1}$ and $M_{n-i+1}$ with is a realization of the random variable $\bs M_{n-i+1}$ such that $\hat{\sigma}_{n-i+1}^2=\sigma_{n-i+1}^2\cdot M_{n-i+1}$.\\
\\
We deduce that
\begin{align*}
&\bs Z^{model}_{adj}\sim \sum\limits_{i=2}^n (f_{n-i+1}-1)\cdot C_{i,n-i}+\sqrt{\frac{\sum\limits_{i=2}^n \sigma_{n-i+1}^2 C_{i,n-i}^2}{\sum_{l=0}^{i-1} C_{l,n-i}^\gamma}}\cdot\tilde{\zeta}\\
&\scriptstyle+\sqrt{\frac{\sum\limits_{i=2}^n \sigma_{n-i+1}^2 C_{i,n-i}^2\left(\frac{1}{C_{i,n-i}^\gamma}+\frac{1}{\sum_{l=0}^{i-1} {C_{l,n-i}^\gamma}}\right) \cdot M_{n-i+1}\cdot \sum\limits_{i=2}^n \sigma_{n-i+1}^2 C_{i,n-i}^2\left(\frac{1}{C_{i,n-i}^\gamma}+\frac{1}{\sum_{l=0}^{i-1} C_{l,n-i}^\gamma}\right)}{\sum\limits_{i=2}^n \sigma_{n-i+1}^2 C_{i,n-i}^2 \left(\frac{1}{C_{i,n-i}^\gamma}+\frac{1}{\sum_{l=0}^{i-1} C_{l,n-i}^\gamma}\right)\bs M_{n-i+1}^\prime}}\cdot\bs\zeta^\prime
\end{align*}
where $\tilde{\zeta}:=\sum_i \sigma_{n-i+1}\cdot R_{n-i+1}\cdot C_{i,n-i}\cdot \left(\sqrt{\frac{\sum_i \sigma_{n-i+1}^2 C_{i,n-i}^2}{\sum_l C_{l,n-i}^\gamma}}\right)^{-1}$ is a realization of a standard normally distributed random variable $\tilde{\bs \zeta}$ independent of both, $\bs \zeta^\prime$ and $\bs M_i$ for all $i$.\\
\\
Set $G(\mathcal{R}):=F_{\bs Z^{model}_{adj}(\mathcal{R})} (\bs Z)$ where 
$\bs Z=\sum_{i=1}^n \bs Z_{i,n-i+1}$ and consider the random variable $G(\mathcal{\bs R})$. With an independent, standard normally distributed random variable $\bs \zeta$ and using some algebraic manipulation exploiting properties of the normal distribution we derive
\begin{align*}
&G(\mathcal{\bs R})=F_{\bs Z^{model}_{adj}}\left(\sum\limits_{i=2}^n (f_{n-i+1}-1)C_{i,n-i}+\sigma_{n-i+1}^2 \sqrt{C_{i,n-i}^{2-\gamma}}\cdot\bs \zeta_{i,n-i+1}\right)\\
&\scriptsize\sim F_{\sqrt{\frac{\sum\limits_{i=2}^n \sigma_{n-i+1}^2 C_{i,n-i}^2\left(\frac{1}{C_{i,n-i}^\gamma}+\frac{1}{\sum_{l=0}^{i-1} {C_{l,n-i}^\gamma}}\right)}{\sum\limits_{i=2}^n \sigma_{n-i+1}^2 C_{i,n-i}^2\left(\frac{1}{C_{i,n-i}^\gamma}+\frac{1}{\sum_{l=0}^{i-1} {C_{l,n-i}^\gamma}}\right)\bs M_{n-i+1}^\prime}
	}\cdot \bs \zeta^\prime}\left(\frac{-\sqrt{\frac{\sigma_{n-i+1}^2C_{i,n-i}^2}{\sum_{l=0}^{i-1} C_{l,n-i}^\gamma}}\cdot \tilde{\bs \zeta}+\sqrt{\sum\limits_{i=2}^n \sigma_{n-i+1}^2 C_{i,n-i}^{2-\gamma}}\cdot \bs \zeta}{\sqrt{\sum\limits_{i=2}^n \sigma_{n-i+1}^2 C_{i,n-i}^2\left(\frac{1}{C_{i,n-i}^\gamma}+\frac{1}{\sum_{l=0}^{i-1} C_{l,n-i}^\gamma}\right)\bs M_{n-i+1}}}\right)\\
&\scriptsize\sim F_{\sqrt{\frac{\sum\limits_{i=2}^n \sigma_{n-i+1}^2 C_{i,n-i}^2\left(\frac{1}{C_{i,n-i}^\gamma}+\frac{1}{\sum_{l=0}^{i-1} {C_{l,n-i}^\gamma}}\right)}{\sum\limits_{i=2}^n \sigma_{n-i+1}^2 C_{i,n-i}^2\left(\frac{1}{C_{i,n-i}^\gamma}+\frac{1}{\sum_{l=0}^{i-1} {C_{l,n-i}^\gamma}}\right)\bs M_{n-i+1}^\prime}}\cdot \bs \zeta^\prime}\left(\sqrt{\frac{\sum\limits_{i=2}^n \sigma_{n-i+1}^2C_{i,n-i}^2\left(\frac{1}{C_{i,n-i}^\gamma}+\frac{1}{\sum_{l=0}^{i-1} {C_{l,n-i}^\gamma}}\right)}{\sum\limits_{i=2}^n \sigma_{n-i+1}^2 C_{i,n-i}^2\left(\frac{1}{C_{i,n-i}^\gamma}+\frac{1}{\sum_{l=0}^{i-1} C_{l,n-i}^\gamma}\right)\bs M_{n-i+1}}}\cdot\bs \zeta \right)\\
&\scriptsize\sim F_{\frac{\bs \zeta^\prime}{\left(\sum\limits_{i=2}^n \sigma_{n-i+1}^2 C_{i,n-i}^2\left(\frac{1}{C_{i,n-i}^\gamma}+\frac{1}{\sum_{l=0}^{i-1} {C_{l,n-i}^\gamma}}\right)\bs M_{n-i+1}^\prime\right)^{\frac{1}{2}}}}
\left(\frac{\bs \zeta}{\left(\sum\limits_{i=2}^n \sigma_{n-i+1}^2 C_{i,n-i}^2\left(\frac{1}{C_{i,n-i}^\gamma}+\frac{1}{\sum_{l=0}^{i-1} C_{l,n-i}^\gamma}\right)\bs M_{n-i+1}
\right)^{\frac{1}{2}}}\right).
\end{align*}
Hence, $G(\mathcal{\bs R})$ is uniformly distributed. The assertion of the theorem follows from
$$
P\left(\sum_{i=1}^n \bs Z_{i,n-i+1}\le \SCR_{\bs Z}(\alpha;\boldsymbol{\mathcal{R}};M)\right)=P\left(G(\mathcal{\bs R})\le \alpha \right)=\alpha.
$$
\end{proof}
\end{appendix}

\section{General procedure in Section 4}\label{sec:app2}
\subsection{Example}\label{subsec:data}
Throughout Section \ref{sec:exmethods} we use the following example: Consider a claims development triangle with $n=10$. We assume that the starting values $C_{i,0}$ are normally distributed with mean $f_0=1,420,000$ EUR and standard deviation $\sigma_0=336,000$ EUR.\\ 
The ``true'' parameters $f_k$ and $\sigma_k$ are given by
\begin{table}[H]
	\begin{center}
		\begin{tabular}{|c|c|c|c|c|c|c|c|c|c|c|c|}
			\hline
			&\multicolumn{10}{|c|} {Development year $k$}\\
			\hline
			\hline
			&\small 1 &\small 2 &\small 3 &\small 4 &\small 5&\small 6 &\small 7 &\small 8 &\small 9 &\small 10\\
			\hline
			$f_k$ &\small	1.5&\small	1.2&\small 1.12&\small 1.07&\small 1.04&\small 1.02&\small 1.01&\small 1.005&\small 1.002&\small 1.0\\
			\hline
			$\sigma_k\cdot \sqrt{f_0}^{-\gamma}$ &\small 0.2&\small 0.12&\small 0.08&\small 0.045&\small 0.03&\small 0.018&\small 0.01&\small 0.006&\small 0.003&\small 0.0\\
			\hline
		\end{tabular}
	\end{center}
	\caption{Development factors and their standard deviation}\label{tab:zwei}
\end{table}
Note that we assume all payments to be settled after $10$ development years, i.e. $f_{10}=1$ and $\sigma_{10}=0$.
\subsection{The general procedure to determine the probability of solvency}\label{subsec:general}
To determine the probability of solvency  
\begin{equation}\label{eq:prob}
P(\bs X\le \SCR(99.5\%;\bs D;M))
\end{equation}
experimentally we use the following general procedure based on a Monte-Carlo simulation. Fixing the ``true'' parameters $f_k$, $\sigma_k$ we run through the following steps:
\begin{enumerate}
	\item (Outer loop over $s$ different random triangles) Using the normal model assumptions and given the parameters $f_k$ and $\sigma_k$, $0\le k\le n$, draw $s$ different random development triangles $D_j$, $1\le j\le s$.
	\item For each triangle $D_j$, $1\le j\le s$, do the following:
	\begin{enumerate}
		\item (Simulation of the ``true'' claims development result $\bs X$) Using the normal model assumptions and the parameters $f_k$ and $\sigma_k$ draw random realizations of $\bs Z_{i,n-i+1}$, $1\le i\le n$, representing the payments of the next business year. Estimating $\hat{R}_0(D_j)$ and $\hat{R}_1(D_j;\overrightarrow{\bs Z})$ as described in Subsection \ref{subsec:grand} using the deterministic chain-ladder method we get a realization $x_j$ of 
		$$
		\bs X=\bs X(D_j, \overrightarrow{\bs Z})=\hat{R}_1(D_j,\overrightarrow{\bs Z})+\bs Z-\hat{R}_0(D_j)
		$$ 
		representing the claims development loss of the next business year.
		\item (Determination of the risk capital) Independently of $x_j$ we then determine the SCR using a Monte-Carlo simulation with $t$ scenarios. For each of the $t$ scenarios we draw a realization from $\bs X^{model}(D_j)$ where 
		\begin{align*}
		\bs X^{model}(D_j):=\left\{\begin{array}{ll}
		\bs X^{model}_{without}(D_j) &\mbox{ in Subsection \ref{subsec:without},}\\
		\bs X^{model}_{BT}(D_j) &\mbox{ in Subsection \ref{subsec:boots},}\\
		\hat{\bs X}^{model}_{adj}(D_j) &\mbox{ in Subsection \ref{subsec:inv}.}\\
		\end{array}
		\right.
		\end{align*}
		We set the solvency capital requirement $\SCR=\SCR(D_j)$ equal to the empirical $\alpha$-quantile determined by the simulation described above. It approximates the quantile $F^{-1}_{\bs X^{model}}(\alpha).$
		\item (Does the risk capital cover the loss?) We compare the realization $x_j$ with $\SCR=\SCR(D_j)$.
	\end{enumerate}
	\item (Determination of the probability of solvency) Count how many times we observe $x_j\le \SCR(D_j)$. The relative frequency approximates the probability (\ref{eq:prob}). 
\end{enumerate}
For the calculations of the results presented in Section \ref{sec:exmethods} we used $s=100,000$ and $t=10,000$ simulations.
\begin{rem}
	For the normal model it is theoretically possible that the chain ladder development factors become negative resulting in negative cumulative claims. In the rare cases where we observed negative factors we reset the factor equal to 1.0. Note that small factors correspond to small realizations of $\bs S$ which do not effect the probability of solvency focusing on large realizations.
\end{rem}

\medskip
\small \textbf{Andreas Fr\"ohlich}, Zentrales Aktuariat Komposit, R+V Allgemeine Versicherung AG, Raiffeisenplatz 1,
65189 Wiesbaden, Germany, Email: \url{andreas.froehlich@ruv.de}\\
\\
\textbf{Annegret Weng (corresponding author)}, Hochschule f\"ur Technik, 
Schellingstr. 24. 70174 Stuttgart, Germany,
Tel.: 0049-(0)711- 8926-2730, Fax: 0049-(0)711-8926-2553,
Email: \url{annegret.weng@hft-stuttgart.de}


\begin{thebibliography}{}
	
	\bibitem[{Barndorff-Nielsen} and {Cox}, 1996]{barndorff}
	{Barndorff-Nielsen}, O. and {Cox}, C. (1996).
	\newblock Prediction and asymptotics.
	\newblock {\em Bernoulli}, 2(4):319--340.
	
	\bibitem[{Bignozzi} and {Tsanakas}, 2016a]{bignozzi}
	{Bignozzi}, V. and {Tsanakas}, A. (2016a).
	\newblock Model uncertainty in risk capital measurement.
	\newblock {\em Journal of Risk}, 18(3):1--24.
	
	\bibitem[{Bignozzi} and {Tsanakas}, 2016b]{bignozzi2}
	{Bignozzi}, V. and {Tsanakas}, A. (2016b).
	\newblock Parameter uncertainty and residual estimation risk.
	\newblock {\em Journal of Risk and Insurance}, 83(4):949--978.
	
	\bibitem[{Bj\"orkwall}, 2011]{bjoerkwall2011}
	{Bj\"orkwall}, S. (2011).
	\newblock Stochastic claims reserving in non-life insurance.
	\newblock available under
	\url{http://su.diva-portal.org/smash/get/diva2:406884/FULLTEXT01.pdf},
	retrieved on 31/08/2016.
	
	\bibitem[{Diers} and {Kraus}, 2010]{krause2009}
	{Diers}, D. and {Kraus}, C. (2010).
	\newblock Das stochastische {R}e-{R}eserving - {E}in simulationsbasierter
	{A}nsatz f{\"u}r die stochastische {M}odellierung des {R}eserverisikos in der
	{K}alenderjahressicht.
	\newblock {\em Zeitschrift f{\"u}r die gesamte Versicherungswissenschaft},
	99(1):41--64.
	
	\bibitem[{England} and {Verrall}, 1999]{england1999}
	{England}, P. and {Verrall}, R. (1999).
	\newblock Analytic und bootstrap estimates of prediction errors in claims
	reserving.
	\newblock {\em Insurance: Math. Econom.}, 25(3):281--293.
	
	\bibitem[{England} and {Verrall}, 2006]{england2006}
	{England}, P. and {Verrall}, R. (2006).
	\newblock Predictive distributions of outstanding liabilities in general
	insurance.
	\newblock {\em Annals of Actuarial Science}, 1(2):221--270.
	
	\bibitem[{Fisher}, 1930]{fisher1930}
	{Fisher}, R. (1930).
	\newblock Inverse probability.
	\newblock {\em Proceedings of the Cambridge Philosophical Society},
	26:528--535.
	
	\bibitem[{Fr\"ohlich} and {Weng}, 2015]{froehlich2014}
	{Fr\"ohlich}, A. and {Weng}, A. (2015).
	\newblock Modelling parameter uncertainty for risk capital calculation.
	\newblock {\em European Actuarial Journal}, 5(1):79--112.
	
	\bibitem[{Fr\"ohlich} and {Weng}, 2017]{froehlich2015}
	{Fr\"ohlich}, A. and {Weng}, A. (2017).
	\newblock Parameter uncertainty for risk capital calculations for aggregated
	insurance portfolios.
	\newblock Available under
	\url{http://versicherung-mathematik.de/wp-content/uploads/2017/03/Article_Multinormal_preprint.pdf},
	retrieved on 18/09/2016.
	
	\bibitem[{Gerrard} and {Tsanakas}, 2011]{gerrard}
	{Gerrard}, R. and {Tsanakas}, A. (2011).
	\newblock Failure probability under parameter uncertainty.
	\newblock {\em Risk Analysis}, 8(5):727--744.
	
	\bibitem[{Gisler}, 2006]{gisler2006}
	{Gisler}, A. (2006).
	\newblock The estimation error in the chain-ladder reserving method: a bayesian
	approach.
	\newblock {\em ASTIN Bulletin}, 36:554--565.
	
	\bibitem[{Hannig}, 2013]{hannig2013}
	{Hannig}, J. (2013).
	\newblock Generalized finducial inference via discretization.
	\newblock {\em Statistica Sinica}, 23:489--514.
	
	\bibitem[{Hannig} et~al., 2016]{hannig2016}
	{Hannig}, J., {Iyer}, H., {Lai}, R., and {Lee}, T. (2016).
	\newblock Generalized fiducial inference: A review and new results.
	\newblock {\em Journal of the American Statistical Association},
	111:1346--1361.
	
	\bibitem[{Hannig} et~al., 2006]{hannig2006}
	{Hannig}, J., {Iyer}, H., and {Patterson}, P. (2006).
	\newblock Fiducial generalized confidence intervals.
	\newblock {\em Journal of the American Statistical Association}, 101:254--269.
	
	\bibitem[{Hora} and Buehler, 1966]{hora}
	{Hora}, R. and Buehler, R. (1966).
	\newblock Fiducial theory and invariant estimation.
	\newblock {\em Annals of Mathematical Statistics}, (37):643--656.
	
	\bibitem[{Iyer} et~al., 2004]{iyer2004}
	{Iyer}, H., {Wang}, C., and {Mathew}, T. (2004).
	\newblock Models and confidence intervals for true values in interlaboratory
	trials.
	\newblock {\em Journal of the American Statistical Association}, 99:1060--1071.
	
	\bibitem[{Mack}, 1993]{Mack1993}
	{Mack}, T. (1993).
	\newblock Distribution-free calculation of the standard error of chain ladder
	reserve estimates.
	\newblock {\em ASTIN Bulletin}, 23(2):213--225.
	
	\bibitem[{Mack}, 1999]{Mack1999}
	{Mack}, T. (1999).
	\newblock The standard error of chain ladder reserve estimates: recursive
	calculation und inclusion of a tail factor.
	\newblock {\em ASTIN Bulletin}, 29(2):361--366.
	
	\bibitem[{Merz} and {W\"uthrich}, 2008]{merz2008}
	{Merz}, M. and {W\"uthrich}, M. (2008).
	\newblock Modelling the claims development result for solvency purposes.
	\newblock CAS E-Forum.
	
	\bibitem[{Peters} et~al., 2017]{peters}
	{Peters}, G., {Targino}, R., and {W\"uthrich}, M. (2017).
	\newblock Full bayesian analysis of claims reserving uncertainty.
	\newblock {\em Insurance: Mathematics and Economics}, 72:95--106.
	
	\bibitem[{Pinheiro} et~al., 2003]{pinheiro2003}
	{Pinheiro}, P., {Andrade e Silva}, J., and {de Lourdes Centeno}, M. (2003).
	\newblock Bootstrapping methodology in claim reserving.
	\newblock {\em J. Risk Insurance}, 70(4):701--714.
	
	\bibitem[{Pitera} and {Schmidt}, 2016]{pitera}
	{Pitera}, M. and {Schmidt}, T. (2016).
	\newblock Unbiased estimation of risk.
	\newblock available under \url{https://arxiv.org/abs/1603.02615}.
	
	\bibitem[{Severini} et~al., 2002]{serverini}
	{Severini}, T., {Mukerjee}, R., and {Gosh}, M. (2002).
	\newblock On an exact probability matching property of right-invariant priors.
	\newblock {\em Biometrika}, 89:952--957.
	
	\bibitem[{Solvency directive 2009/138/EC}, 2009]{solvency}
	{Solvency directive 2009/138/EC} (2009).
	\newblock approved by the Eureopan Parliament and the Council of the European
	Union, available under
	\url{http://eur-lex.europa.eu/LexUriServ/LexUriServ.do?uri=OJ:L:2009:335:0001:0155:en:PDF},
	retrieved on 31/08/2016.
	
	\bibitem[{Verrall}, 1990]{verrall1990}
	{Verrall}, R. (1990).
	\newblock Bayes und empirical bayes estimation for the chain ladder model.
	\newblock {\em ASTIN Bulletin}, 20(2):217--243.
	
	\bibitem[{Wang} et~al., 2012]{wang2012}
	{Wang}, J.-M., {Hannig}, J., and {Iyer}, H. (2012).
	\newblock Fiducial prediction intervals.
	\newblock {\em Journal of Statistical Planning and Inference}, 142:1980--1990.
	
	\bibitem[{W\"uthrich} and {Merz}, 2008]{wuethrich2008}
	{W\"uthrich}, M. and {Merz}, M. (2008).
	\newblock {\em Stochastic claims reserving methods in insurance}.
	\newblock John Wiley \& Sons, Ltd.
	
	\bibitem[{Young} and {Smith}, 2005]{young}
	{Young}, G. and {Smith}, R. (2005).
	\newblock {\em Essentials of statistical inference}.
	\newblock Cambridge University Press.
	
	\bibitem[{Zabell}, 1992]{zabell92}
	{Zabell}, S. (1992).
	\newblock {R.}{A.} {F}isher and fiducial argument.
	\newblock {\em Statistical Science}, 7(3):369--387.
	
\end{thebibliography}
\end{document}